\documentclass[a4paper,12pt]{article}
\usepackage{graphicx}
\usepackage{multirow}
\usepackage{bbm}
\usepackage{float}
\usepackage{footnote}
\usepackage{amssymb}
\usepackage{pifont}
\usepackage{ifpdf}
\ifpdf
\else
\usepackage{pstcol,pst-fill,pstricks}
\fi
\usepackage{natbib}
\usepackage{mathpazo}
\usepackage{import}
\usepackage{dsfont}
\usepackage{enumerate}
\usepackage{a4}
\usepackage{lscape}
\usepackage{epsfig}
\usepackage[latin1]{inputenc}
\usepackage[OT1]{fontenc}
\usepackage[T1]{fontenc}
\usepackage{color}
\usepackage[hyperfootnotes=false]{hyperref}
 \hypersetup{colorlinks=true,
 linkcolor=blue,
 citecolor=blue,
 filecolor=black,
 urlcolor=black,
 pdfstartview={Fit},
pdfpagemode=UseNone
 }
\usepackage{geometry}
\usepackage{mdframed}
\usepackage{tikz}
\usepackage[multiple]{footmisc}

\makeatletter
\newcommand*\bigcdot{\mathpalette\bigcdot@{.5}}
\newcommand*\bigcdot@[2]{\mathbin{\vcenter{\hbox{\scalebox{#2}{$\m@th#1\bullet$}}}}}
\makeatother

\usepackage{setspace}
 \newcommand{\be}{\begin{equation}}
\newcommand{\ee}{\end{equation}}
\newcommand{\E}{\mathbb{E}}
\usepackage{mathtools}
\usepackage{amsmath}
\usepackage{amsthm}
\usepackage{thmtools}
\usepackage{calrsfs}
\DeclareMathAlphabet{\pazocal}{OMS}{zplm}{m}{n}
\allowdisplaybreaks

\declaretheoremstyle[
  headfont=\normalfont\scshape,
  numbered=unless unique,
  bodyfont=\normalfont,
  spaceabove=1em,
  prefoothook=\newline\rule{\linewidth}{1pt},
  spacebelow=1em,
]{exmpstyle}

\declaretheorem[
  style=exmpstyle,
  title=\textbf{Definition},
  refname={example,examples},
  Refname={Example,Examples}
]{ddef}

\DeclareMathOperator{\Tr}{Tr}
\newtheorem{Ass}{Assumption}
\newtheorem{prop}{Proposition}

\theoremstyle{definition}

\newtheorem{theorem}{Theorem}
\newtheorem{corollary}{Corollary}[theorem]

\begin{document}
\begin{titlepage}

\title{Analyzing Linear DSGE models: the Method of Undetermined Markov States\thanks{\scriptsize{I thank He Nie and Zheng Zhongxi for excellent research assistance. I thank Alexis Akira-Toda, Adrien Bilal, Giovanni Caggiano, Fabrice Collard, Gregory Cox, Thanasis Geromichalos, Edward Herbst, Alisdair McKay and Denis Tkachenko for useful comments/discussions as well as seminar participants at Monash University, National University of Singapore, the 2021 $\&$ 2022 Computing in Economics and Finance (CEF) and International Association for Applied Econometrics (IAAE) conferences as well as the Dynare 2022 conference. I also want to thank researchers outside of the fields of economics that where kind enough to provide comments and clarifications on the results presented in this paper. That includes Jon Chapman, Radek Erban, Joao Hespanha and Alejandro Ribeiro. Finally, I thank an anonymous referee for interesting suggestions that greatly improved the paper.}}}
\author{Jordan Roulleau-Pasdeloup\thanks{\scriptsize{Department of Economics, National University of Singapore
1 Arts Link, AS2 \# 05-22 - Singapore 117570. Contact: jordan.roulleau@gmail.com}}}
\date{\vspace{-.25cm}\today\vspace{-1.2cm}}

\maketitle
\thispagestyle{empty}
\abstract{I show that a class of Linear DSGE models with one endogenous
state variable can be represented as a three-state Markov chain. I
develop a new analytical solution method based on this representation, which
amounts to solving for a vector of Markov states and one transition
probability. These two objects constitute sufficient statistics to compute in closed form objects that have routinely been computed numerically: impulse
response function, cumulative sum, present discount value multiplier. I apply the method to a standard New Keynesian model that features optimal monetary policy with commitment.}\\[.5cm]
\end{titlepage}
\setlength{\parskip}{1em}

\section{Introduction}

Most macroeconomic models feature (rational) expectations as well as endogenous persistence. In these models, the effects of exogenous shocks propagate through their internal transmission mechanisms. These can be contrasted with models that feature no endogenous persistence, which inherit the persistence of exogenous shocks.\footnote{The standard log-linear New Keynesian model described in \cite{Gali2015monetary} fits into this category.} The flip side of the coin is that models without endogenous persistence can be solved in closed form and easily represented graphically. In this context, adding endogenous persistence comes at a price: while it makes the model more interesting, the model is usually not amenable anymore to a closed-form solution or a simple graphical representation. 

Given this difficulty, the usual procedure is to start from a special version of the model without endogenous persistence to present the intuition and then move on to the model with endogenous persistence. The latter is then studied through numerical results.\footnote{See \cite{Christiano2011} for a New Keynesian model applied to fiscal policy at the ZLB. See also \cite{Sims2019four} for a recent contribution that looks at the interaction of heterogeneous households, banks and quantitative easing.} There are two main drawbacks associated with this procedure. The first is that one can never be sure that the simple model is entirely faithful to the more elaborated one: endogenous persistence can introduce entirely new transmission mechanisms. Second, several issues can only be studied in models with endogenous persistence.\footnote{For example, durable goods (\cite{Barsky2007sticky}), public debt (\cite{Bianchi2017escaping}), search and matching frictions on the labor market (\cite{Michaillat2012matching}), firm dynamics (\cite{Bilbiie2012endogenous}), investment (\cite{Mcgrattan2012capital}) and so and so on.} In these cases, looking at the model without endogenous persistence is generally not an option.

In this paper, I develop a set of tools to help overcome this trade-off. These tools aim to facilitate the interpretation of the effects of shocks and policies in linear DSGE models using closed form results. The main difficulty in a wide class of linear DSGE models comes from the interplay of expectations and endogenous persistence. Say one has a linear model with output $y_t$ and an endogenous state variable $k_t$ whose law of motion explicitly depends on $y_t$ \textemdash the standard RBC model where $k_t$ is capital is one example. This model also features an exogenous process $z_t$ that follows a standard $AR(1)$ process. What makes models without endogenous persistence amenable to closed form solutions is that one can compute the conditional expectation of output next period as a linear function of current output: there exists a constant $c$ such that $\E_t y_{t+1} = c\cdot y_t$. Likewise, what makes models \textit{with} endogenous persistence complicated is that it is not possible to express this conditional expectation as a function of $y_t$ anymore.

The tools that I develop in this paper allow one to compute expectations in a more intuitive way in this context. This is achieved by representing the dynamics of $y_t$ using a finite state, forward looking Markov chain $\pazocal{Y}_t$. For the purpose of this paper, this Markov chain will effectively take on two states $y_I$ and $y_M$ where 'I' stands for impact and 'M' for medium run.\footnote{Strictly speaking, the chain will visit a third state which is the absorbing steady state. Since I am focusing on models expressed in log deviations from steady state this third state will be equal to zero.} The Markov chain structure implies that one only has to care about these two states going forward so that expectations can be computed in a straightforward manner. Loosely speaking, the method guarantees that there exists constants $c_1$ and $c_2$ such that $\E_I \pazocal{Y}_{t+1} = c_1\cdot y_I$ and $\E_M \pazocal{Y}_{t+1} = c_2\cdot y_M$, where $E_I$ denotes expectations conditional on being in state 'I' and likewise for state 'M'. In the special case of no endogenous persistence $c_2$ is irrelevant and $c_1$ boils down to the constant $c$ introduced before. For this reason, the Markov structure allows one to consider endogenous persistence while retaining the analytical tractability and nesting the case without endogenous persistence.

More generally, I show that a class of linear DSGE models with one endogenous state variable have a common structure which can be \emph{exactly} replicated using a simple, three-state Markov chain.\footnote{Models without endogenous persistence that assume two-state Markov chains for exogenous shocks arise naturally as a special case of the approach developed in this paper. See for example \textemdash \cite{Eggertsson2010}, \cite{Christiano2011} and \cite{Woodford2011}.} Leveraging the structure of the solution, I develop a new method to solve these models in closed form that consists in solving for the relevant states and transition probabilities of the Markov chain: the Method of Undetermined Markov States. I show that the method is straightforward to implement in practice. The key result underpinning the method is that a three-state Markov chain can replicate the hump-shaped behavior typical of endogenous states in DSGE models. In Online Appendix B, I show under which conditions the method generalizes to an arbitrary number of endogenous states. As long as the transition matrix in the state space representation of the model does not feature any imaginary eigenvalues, then the method still applies. For this reason, I will focus on the equilibrium concept developed in \cite{Mccallum1983non}, which explicitly rules these out.\footnote{See \cite{Mccallum1983non}, footnote 9.}

This might sound counter-intuitive at first since each run of such a Markov chain is a step function. However, by taking the average over arbitrary many realizations of the Markov chain at each point in time\textemdash which is equivalent to computing the conditional expectation, I show  that one can get an exact match. Using this method, it can be shown in a straightforward way that a two-state Markov chain can replicate the impulse response of a typical $AR(1)$ process: maximum effect on impact and then exponential decay to 0. To generate a hump-shape, one just needs an extra state squeezed in-between these two, which will enable the three-state Markov chain to generate a hump-shape \textit{on average}. One can then think of the impulse response of a linear DSGE model as a Markov chain: the initial distribution is the initial shock/impact, while the stationary distribution of the (absorbing) chain is the steady state. Characterising how this Markov chain moves from its initial to its stationary distribution is the main focus of this paper.

The method developed in this paper guarantees that knowledge of a few states and transition probabilities enables one to access virtually all there is to know about the underlying model: these are sufficient statistics. Let us go back to the generic model described earlier and assume that the exogenous process $z_t$ has persistence $p$. The Method of Undetermined Markov States will yield a short run state $y_I$ and a medium run state $y_M$ for output, as well as the endogenous persistence $q$ of the medium run state.
In this context, I show that the impulse response, the cumulative effect on output or the present discount value multiplier can all be written in closed forms as simple functions of $q,y_I,y_M$, and $p$ \textemdash which is given. Given that expectations can be computed as linear functions of the controls, the method can be used to plot the exact first order conditions of the model graphically with endogenous expectations explicitly taken into account. 

To illustrate the method, I develop a closed form solution of a standard New Keynesian model where the Central Bank behaves optimally with full commitment. I show that the model can be represented graphically by a pair of aggregate supply/targeting rule relations: one for the short run and one for the medium run. 
For the short run, I show that the logic described in \cite{Mcleay2020optimal} operates: cost-push shocks make the Phillips curve slide along a downward sloping targeting rule. In the medium run however, the story is different. The presence of commitment appears as a shift to the targeting rule so that it now slides along a positive sloping Phillips curve. As a result, the model clarifies that the positive correlation between inflation and the output gap re-appears in the medium run after being masked in the short run. 

This paper is structured as follows. I show that a class of LRE models can be solved using a three-state Markov chain representation in Section \ref{sec:baseline_model}. I briefly review related approaches in Section \ref{sec:related}. In Section \ref{sec:analytical_tools}, I show that relevant model objects can be computed analytically as a function of a few Markov states/transition probabilities. In Section \ref{sec:application}, I use the method to solve a standard New Keynesian model with optimal monetary policy. I conclude in Section \ref{sec:conclusion}.

\section{The Method of Undetermined Markov States}
\label{sec:baseline_model}

In this section, I will work with a class of generic linear DSGE models and describe a new method that allows one to simply compute their impulse response function (henceforth IRF) in closed form. The reason for the focus on the IRF is that, for a linear model the IRF summarizes all there is to know about it. Indeed, a standard result in \textit{Linear Systems Theory} is that any path of the endogenous variables can be expressed as a linear combination of the IRFs at each point in time.\footnote{See for example Theorem 3.2 in \cite{Hespanha2018linear}.} For simplicity of exposition, I will consider a class of models that feature an arbitrary number of control/forward-looking variables, one endogenous state and one exogenous state \textemdash the generalization to an arbitrary number of endogenous states is covered in the online Appendix B. More precisely, throughout this section I will work with the following class of Linear DSGE models:

\begin{ddef}[Model $\mathcal{A}$]
\label{def:model_AM}
I consider a class of models\footnote{This class of models can be easily extended to include $\E_t k_{t+1}$ in equation \eqref{eq:fwd_multi} or $\mathbb{Y}_{t-1}$ on the right hand side of equation \eqref{eq:bwd_multi}.} that relate a vector of control/forward-looking variables $\mathbb{Y}_t$ to an endogenous state $k_t$ and an exogenous state $z_t$ through the following set of equations:
\begin{align}
\label{eq:fwd_multi}
\mathbf{A}_0\mathbb{Y}_t &= \mathbf{A}_1\E_t \mathbb{Y}_{t+1}+B_0 k_t+C_0 z_t\\
\label{eq:bwd_multi}
k_t &= \rho k_{t-1} + D_0 \mathbb{Y}_t+ez_t\\
z_t &= pz_{t-1}+\epsilon_t,
\label{eq:exo_state_multi}
\end{align}
where $t\in \mathbb{N}$, $\epsilon_t\sim i.i.d\ \pazocal{N}(0,\sigma)$, $k_{-1}=z_{-1}=0$ are given and $\E_t \equiv\E(\cdot|\Omega_t)$ where $\Omega_t$ denotes the information set at time $t$. The endogenous and exogenous states $k_t$ and $z_t$ are scalars while $\mathbb{Y}_t = \left[y_{1,t}\cdots\ y_{N,t}\right]^\top$ for $N\in\mathbb{N}_+$, $N\geq 1$ and all vectors/matrices are conformable.
\end{ddef}
This encompasses a wide array of models where the endogenous state can react on impact. A prime example would be the canonical RBC model, which is studied in detail in Online Appendix A. This would also be the case in a model where monetary policy is set optimally and $k_t$ represents the Lagrange multiplier on the Phillips curve constraint. I study this model in Section \ref{sec:application}.

The class of models considered here is labelled $\mathcal{A}$ because both the endogenous and exogenous states are auto-regressive. I restrict the parameter for exogenous persistence so that $0\leq p<1$. I do not impose any restrictions on the parameter $\rho$ that governs endogenous persistence, but models used in macroeconomics are typically characterized by $0\leq\rho\leq 1$. A typical solution for model $\mathcal{A}$ takes the following form:
\begin{ddef}[Model $\mathcal{A}$: Solution]
\label{def:Solution_A}
A Minimum State Variable (MSV) solution of model $\mathcal{A}$ can be written in state-space form as:
\begin{align}
\label{eq:state_space_k}
k_t &= \eta_{kk}k_{t-1}+\eta_{kz}z_{t}\\
y_{i,t} &= \eta^{(i)}_{yk}k_{t-1}+\eta^{(i)}_{yz}z_{t},
\label{eq:state_space_y}
\end{align}
for $i \in \left\{1,\dots, N\right\}$. The first equation is the state equation and the second is the measurement equation. 
\end{ddef}
To ensure that such a solution exists, I will maintain the following assumption throughout this paper:
 \begin{Ass}
\label{ass:BK}
Matrices $\mathbf{A}_0,\mathbf{A}_1$, vectors $B_0, B_1, D_0, D_1$ and the scalar $\rho$ are such that Model $\mathcal{A}$ meets the conditions for a unique MSV equilibrium in \cite{Mccallum1983non}.
\end{Ass}

There are several reasons for focusing on the unique MSV equilibrium. First, such an equilibrium has been shown to be E-stable in \cite{Mccallum2003unique}. Second, any model belonging to class $\mathcal{A}$ that satisfies the \cite{Blanchard1980solution} conditions has a solution that coincides with the MSV.\footnote{This may not be the case however if there are more than one endogenous state variable. See discussion in the Online Appendix B.} In some cases, if the \cite{Blanchard1980solution} conditions are not met so that the equilibrium is indeterminate, one can still construct a unique MSV solution. As a result, the MSV solution concept is more general. Third, recent contributions such as \cite{Angeletos2021determinacy} show that the presence of indeterminacy may be an artifact of the assumption of perfect transmission of knowledge: if the latter is only slightly imperfect, then only the MSV solution emerges. 

With this in mind, the main contribution of this paper is to recast the dynamics of both $y_{i,t}$ and $k_t$ in terms of a stochastic process that is explicitly forward looking: a finite-state Markov chain.    
 
\subsection{A Markov chain version of model $\mathcal{A}$}
\label{sec:MC_A}

The goal of this section is to detail how one can recast model $\mathcal{A}$ in terms of Markov states and transition probabilities. Remember that computing one period ahead conditional expectations in this context involves computing one for state `I' and one for state `M'. Likewise, the forward and backward equations \eqref{eq:fwd_multi}-\eqref{eq:bwd_multi} will be broken down in two versions for each state of the Markov chain. 

The aim is to have Markov versions of these equations that are equivalent. These new, Markovian equations will in turn imply various relationships between Markov states and transition probabilities. A solution of the model will then be a vector of states and transition probabilities that make all these new equations hold with equality. The specific form that these equations must take is described in the following definition and in more detail in the subsequent discussion.
\begin{ddef}[Model $\mathcal{M}$]
\label{def:model_M_multi}
Let me define a vector-valued Markov chain $\pazocal{M}_t$ on a countable set $\mathbb{S}_M$ with the banded transition matrix:
\begin{align*}
\mathcal{P}_b=
\begin{bmatrix}
p & 1-p & 0\\
0 & q & 1-q\\
0 & 0 & 1
\end{bmatrix}.
\end{align*}
The Markov chain takes on three vectors:
\begin{align*}
\pazocal{M}_I &= \left[1\ k_I\ y_{1,I}\dots\ y_{N,I}\right]^\top\\
\pazocal{M}_M &= \left[0\ k_M\ y_{1,M}\dots\ y_{N,M}\right]^\top\\
\pazocal{M}_L &= \left[0\dots\ 0\right]^\top,
\end{align*}
where $z_I$ has been normalized to 1 and all three vectors are of size $(N+2)\times 1$. The first one is the impact vector. The second one is the medium run vector and the third one is the steady state. Each Markov chain has a degenerate initial distribution such that they all start in the first state. A solution of model $\mathcal{M}$ is a pair of vectors $\pazocal{Y}_I=\left[ y_{1,I}\dots\ y_{N,I}\right]^\top$ as well as $\pazocal{Y}_M=\left[y_{1,M}\dots\ y_{N,M}\right]^\top$ along with scalars $q,k_I$ and $k_M$ that solve the following system of non-linear equations:
\begin{align}
\label{eq:fwd_I_multi}
\mathbf{A}_0\pazocal{Y}_I &= \mathbf{A}_1p\pazocal{Y}_I+(1-p)\mathbf{A}\pazocal{Y}_M + B_0k_I+C_0\\
\label{eq:fwd_M_multi}
\mathbf{A}_0\pazocal{Y}_M &= \mathbf{A}_1q\pazocal{Y}_M +B_0k_M\\
\label{eq:bwd_I_multi}
k_I &= D_0\pazocal{Y}_I+e\\
\label{eq:bwd_M_multi}
k_M &= \frac{\rho}{1-p} k_I + D_0\pazocal{Y}_M\\
\label{eq:kMkI_multi}
k_M &= \frac{qk_I}{1-p}
\end{align}
and is such that $0\leq q<1$. This ensures that the endogenous and exogenous states are perfectly correlated: as soon as the Markov chain jumps from $\pazocal{M}_I$ to $\pazocal{M}_M$, the exogenous state jumps from 1 to 0, the endogenous state jumps from $k_I$ to $k_M$ and each control variable $i$ jumps from $y_{i,I}$ to $y_{i,M}$. Each realization of length $h$ for this Markov chain will yield a $(N+2)\times h$ vector. The first line of this vector will essentially replicate the path for the exogenous state from model $\mathcal{A}$, the second one will replicate the path for the endogenous state and the remaining $N$ will replicate the control variables. As a result, $\mathcal{M}$ implies $N+2$ nested Markov chains: $\pazocal{Z}_t$ for the exogenous state, $\pazocal{K}_t$ for the endogenous state and $\pazocal{Y}_{i,t}$ for each control variable, all with transition matrix $\mathcal{P}_b$. The impulse response for the Markov chain $\pazocal{Y}_{i,t}$ is defined as the sequence $\left\{\E_t \pazocal{Y}_{i,t+n}\right\}_{n\geq 0}$ and likewise for the other two. 
\end{ddef}

Let me now discuss equations \eqref{eq:fwd_I_multi}-\eqref{eq:kMkI_multi} sequentially. The intuition behind equations \eqref{eq:fwd_I_multi}-\eqref{eq:fwd_M_multi} encapsulates why the current Markovian framework will prove to be very simple and amenable to analytical expressions. It uses the fact that one step ahead conditional expectations can be computed in a simple way to rewrite the forward equation \eqref{eq:fwd_multi} in both states `I' and `M'. Take the version for state `I' for example, equation \eqref{eq:fwd_I_multi}. Instead of involving past terms, the forward equation now boils down to a simple relationship between $\pazocal{Y}_I$ and $\pazocal{Y}_M$ for given endogenous/exogenous states. Equation \eqref{eq:fwd_M_multi} for state `M' is even simpler: this is a ``static'' relationship that only involves $\pazocal{Y}_M$ and the endogenous state $k_M$. These two equations are essentially the initial value restrictions which ensure that the Markov chain replicates the forward equation for the first two periods. 

Equations \eqref{eq:bwd_I_multi}-\eqref{eq:bwd_M_multi} follow the same logic. Given that the economy starts from the steady state, the initial variation of the endogenous state can only come from the vector of forward looking variables $\pazocal{Y}_I$. Likewise, equation \eqref{eq:bwd_M_multi} states that $\E_t \pazocal{K}_{t+1} = \rho\E_t\pazocal{K}_{t} + D_0\E_t\pazocal{Y}_{t+1} + e\E_t \pazocal{Z}_{t+1}$ and once again uses the Markov structure to compute the expectations of the Markov chain. It uses \eqref{eq:bwd_I_multi} to cancel out terms involving $\pazocal{Y}_I$ and $e$.

Rewriting the forward- and backward-looking equations of the auto-regressive model thus gives $2N+2$ out of $2N+3$ restrictions so that another one is needed. To derive the last restriction, I am going to use the fact that the $AR(2)$ process for $k_{t+n}$ from the state space solution will imply a precise relationship between $k_M$ and $k_I$. Intuitively, to replicate the hump-shape behavior of the endogenous state, the second state has to be equal to the first one by a factor of $q/(1-p)$. 

The main idea behind the method presented here is that one can break down the dynamics of the economy after a shock into a short and a medium run regime. In turn one can study them one by one. Then, for a given $q$, the actual equilibrium of this economy will be a time-varying linear combination of these two regimes. 

The fact that these Markov chains actually replicate the dynamics of the underlying linear DSGE model comes from focusing on the conditional expectation of each Markov chain. More specifically, the Markov chain structure will imply the following dynamics for these conditional expectations: 
\begin{prop}
\label{prop:p_star}
Let $\pazocal{Y}_{i,t}$ denote a three-state Markov chain with initial distribution $[1,\ 0,\ 0 ]$ and state space $[y_{i,I},\ y_{i,M},\ 0 ]^\top$. Then, the transition matrix $\mathcal{P}_b$ from Definition \ref{def:model_M_multi} implies that the conditional expectation of the Markov chain is given by:
\begin{align*}
\E_t \pazocal{Y}_{i,t+n} &=
\begin{cases}
y_{i,I}\quad\hfill \text{if}\quad n=0\\
p\cdot y_{i,I} + p_{1,2} \cdot y_{i,M}\quad\hfill \text{if}\quad n=1\\
(p+q)\E_t \pazocal{Y}_{i,t+n-1}  - pq\E_t \pazocal{Y}_{i,t+n-2} \quad \hfill   
\text{if}\quad n\geq 2\\
\end{cases}
\end{align*}
\end{prop}
\begin{proof}
See Appendix \ref{sec:proof_prop_p_star}.
\end{proof}
Proposition \ref{prop:p_star} can be viewed as a discrete-time version of a special case of \cite{Kurtz1970solutions}'s theorem, which states that a sequence of pure jump Markov processes converges to a deterministic process. While his proof is based on the convergence of operator semigroups, I provide a much simpler proof based on the Cayley-Hamilton theorem in Appendix \ref{sec:proof_prop_p_star}. 

With this in mind, what Proposition \ref{prop:p_star} says is that after the first two periods the conditional expectation of this Markov chain follows a deterministic process. More precisely, it follows an auto-regressive process of order 2. This helps to explain why the method will replicate the path of endogenous variables obtained with a conventional state-space solution method: with one endogenous state variable and an $AR(1)$ shock, each variable $y_{i,t}$ will follow an $ARMA(2,1)$ process. Viewed through these lenses, the moving average terms will correspond to the first two conditional expectations of the Markov chains for $n=0,1$. Note that these dynamics also hold for the endogenous states as well as the exogenous states.

The result outlined in Proposition \ref{prop:p_star} states that, for a given $q$ the Markov chains will follow these auto-regressive dynamics. What now remains to be shown is that the value of $q$ that obtains from the Markov restrictions \eqref{eq:fwd_I_multi}-\eqref{eq:kMkI_multi} is actually the one that obtains from the procedure outlined in \cite{Mccallum1983non}. This is detailed in the following Theorem:
\begin{theorem}
\label{thm:IRF_univariate}
Assume that Assumption \ref{ass:BK} holds, that the vector $\left[q,k_I,k_M,\pazocal{Y}_I^\top,\pazocal{Y}_M^\top\right]$ solves model $\mathcal{M}$ according to Definition \ref{def:model_M_multi} and that the matrix $\mathbf{M}(z)\equiv\mathbf{A}_0 - z\mathbf{A}_1$ is invertible for $z\in\mathbb{C}$. Then $q$ solves the exact same equation as $\eta_{kk}$ from the state-space solution and the one corresponding to the unique MSV solution can be picked according to the procedure outlined in \cite{Mccallum1983non}. It follows that, after a unit shock at time $t$:
\begin{align}
z_{t+n} = \E_t\pazocal{Z}_{t+n}\\
k_{t+n} = \E_t\pazocal{K}_{t+n}\\
y_{i,t+n} = \E_t\pazocal{Y}_{i,t+n}
\end{align}
for all $i\in\left\{1,\dots,N\right\}$ and $n\geq 0$. 
\end{theorem}
\begin{proof}
See Appendix \ref{sec:Proof_thm_univariate}.
\end{proof}
Theorem \ref{thm:IRF_univariate} establishes that one can use the Markov chain representation from Definition \ref{def:model_M_multi} to solve for the impulse response of model $\mathcal{A}$. The gist of the proof amounts to showing that the sequences $\E_t\pazocal{K}_{t+n}, \E_t\pazocal{Z}_{t+n}$ and $\E_t\pazocal{Y}_{i,t+n}$ for $n\geq 0$ and $i\in\left\{1,\dots,N\right\}$ obey equations \eqref{eq:fwd_multi}-\eqref{eq:exo_state_multi}. A crucial step in this process is that using equations \eqref{eq:fwd_M_multi}, \eqref{eq:bwd_M_multi} and \eqref{eq:kMkI_multi} boil down to a polynomial in $q$ that is equivalent to the one that obtains with \cite{Uhlig1995toolkit}'s undetermined coefficients method. Next, one just needs to use \cite{Mccallum1983non} to select the root that corresponds to the MSV solution. Furthermore, under Assumption \ref{ass:BK}, this solution is unique. In Appendix \ref{sec:Proof_thm_univariate}, I use induction to show that if the Markov chains satisfy equations \eqref{eq:fwd_I_multi}-\eqref{eq:kMkI_multi}, then the conditional expectations of these Markov chains correspond to the MSV solution.

\subsection{A Graphical Illustration}

The fact that a Markov chain with transition matrix $\mathcal{P}_b$ can match the hump-shape dynamics of a given variable $y_{i,t}$ is illustrated in Figure \ref{fig:Visu}. In this figure, I simulate first the dynamics of $y_{i,t}$ after a unit shock $\epsilon_t=1$. In addition, I simulate 50000 runs (indexed by $j=1,\dots,J$) of the three-state Markov chain described in Proposition \ref{prop:p_star} and take the average at each period\textemdash  this is meant solely as a graphical illustration since the conditional expectation has been shown to be computable analytically. 

\begin{figure}
\centering
\caption{Finite Markov chains and ARMA(2,1)}
\includegraphics[width=.7\textwidth]{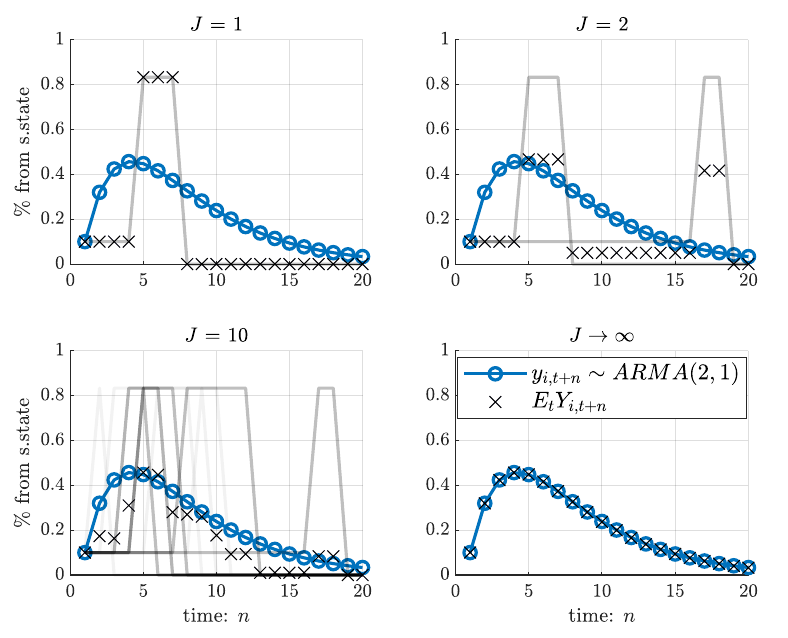}
\begin{minipage}{0.7\textwidth}
\protect\footnotesize Notes: \protect\scriptsize 
The blue circled line plots the deterministic path of a generic $ARMA(2,1)$ process after a unit shock $\epsilon_t=1$. The law of motion for the $ARMA(2,1)$ is given by $y_{i,t} = (p+\eta_{kk})y_{i,t-1} -p\eta_{kk}y_{i,t-2}+\eta^{(i)}_{yz}\epsilon_{t}+(\eta^{(i)}_{yk}\eta^0_{kz}-\eta^{(i)}_{yz}\eta_{kk})\epsilon_{t-1}$. For this particular figure, I use the following parameter values: $\eta_{kk}=0.8$, $\eta^0_{kz}=\eta^{(i)}_{yk}=0.5, \eta^{(i)}_{yz}=0.1$ and $p=0.7$. The exogenous shock is such that $\epsilon_t=1$ and $\epsilon_{t+n}=0$ for $n>0$.
\end{minipage}
\label{fig:Visu}
\end{figure}

For each subplot, I consider a series of realizations of the Markov chain. As before, I assume that the initial distribution is given by the vector $\left[1\ 0\ 0 \right]$ so that the latter starts from its first state. For $J=1$, there is only one realization and thus the average across Markov chain realizations will be the realization itself. One can see that for $J=1$, the specific draw generates a second state that is a bit less persistent compared to the first one. This even more true for the second realization: it only jumps to the second state after 16 periods and stays there for two periods before reaching the third state. This explains why the mean across these two realizations is shifted down in the short run. The lower left panel shows what happens after 10 realizations of the Markov chain. Repeating this process a large number of times, one can see from the bottom-right panel that the average across realizations will converge to the path of the $ARMA(2,1)$. As $J\to\infty$, the average across realizations converges to the conditional expectation of the Markov chain and exactly matches the conditional path of the underlying $ARMA(2,1)$ process. 

\section{Related Approaches}
\label{sec:related}

Given that it deals with linear DSGE models, the approach in the current paper is linked to and complementary with \cite{Blanchard1980solution}, \cite{Anderson1985linear}, \cite{Klein2000using} and \cite{Sims2001}. The closest approach to the one developed here is the method of undetermined coefficients developed in \cite{Uhlig1995toolkit}.

The focus on Markov chains of the current paper is also shared with papers cited earlier that have used these to deal with the Zero Lower Bound (ZLB): \cite{Eggertsson2003}, \cite{Eggertsson2010}, \cite{Christiano2011}, \cite{Mertens2014} and \cite{Nakata2019conservatism} among many others. This framework has permitted a detailed analysis of the characteristics of ZLB-type equilibria in New Keynesian models. In particular, this approach has highlighted that ZLB induced by structural shocks or self-fulfilling prophecies are fundamentally different in their policy implications\textemdash see \cite{Bilbiieneo2022}. These papers have virtually all considered models without endogenous state variables.\footnote{For example, in a recent contribution \cite{Ravn2020macroeconomic} consider a model with search and matching frictions on the labor market. To get rid of the employment rate as an endogenous state variable, they have to assume risk-neutral preferences. \cite{Bhattarai2015time} constitutes one exception as they solve a model with public debt using the Method of Undetermined Coefficients.} As a result, the framework developed in this paper allows one to go further and consider a much wider class of models, all while retaining tractability. I focus on the case of linear models in this paper and leave the model with occasionally binding constraints for future research. 

A series of recent papers have used an absorbing three-state Markov chain to think about monetary policy announcements at the ZLB: \cite{Eggertsson2006}, \cite{Carlstrom2015inflation}, \cite{Bilbiie2019} and \cite{Nie2022}. In contrast with the current paper however, these authors use a transition matrix that is more complicated in that it allows for transitions from state 1 to 3 without visiting state 2. Using the current transition matrix, one can ascertain that there is a finite time $T_M$ when the chain will reach state 2. The solution of the underlying model can then be described as one solution for $[t,t+T_M]$ and one solution for $[t+T_M+1,t+T_L]$ where $T_L$ is the finite date when the chain jumps from state 2 to 3. Beyond the simpler transition matrix, I mainly show that such a three-state Markov chain representation is much more general as it can be used to represent a wide class of linear DSGE models with endogenous persistence.

This paper is also close in spirit to papers that have studied solutions in sequence-space instead of state-space. These include \cite{Boppart2018exploiting}, \cite{Auclert2021using}, \cite{Mckay_Wolf} and \cite{Wolf2021interest}. Instead of solving for, say, the policy rule for aggregate output as a function of the states these papers work directly with the sequence for aggregate output over a finite but long time interval. My approach is similar in that the sequence of aggregate output can be summarized by a finite number of states and transition probabilities. As a result, solving for these Markov states and transition probabilities effectively amounts to solve directly for the sequence of aggregate output.

\section{Analytical tools for Linear DSGE models}
\label{sec:analytical_tools}

There is a long tradition of papers using DSGE models that adopt the following structure: present a simple version of the model analytically, then move on to analyze the full model. Since the full model most likely features endogenous state variables, the workhorse tool to analyze it is usually an impulse response \textemdash see \cite{Christiano2011} for a fairly recent example. In what follows, I will build on Theorem \ref{thm:IRF_univariate} and show that one can use the tools developed in the last section to compute several interesting model objects analytically as well. All of these results follow from Theorem \ref{thm:IRF_univariate} and thus can be considered as corollaries. For each corollary, I will discuss examples where the result can be useful in the context of linear DSGE models.

Since it is the main tool that has been used extensively, I begin with the impulse response function. The following corollary shows that the impulse response of model $\mathcal{M}$ that exactly replicates a model of class $\mathcal{A}$ can be characterized as follows:
\begin{corollary}[Impulse Response Function]
\label{cor:IRF}
The impulse response of the forward-looking variable $\pazocal{Y}_{i,t}$ and the backward-looking variable $\pazocal{K}_t$ at horizon $n$ after a shock $\epsilon_t=1$ are given by:
\begin{align*}
\E_t\pazocal{Y}_{i,t+n} &= p^n y_{i,I} + (1-p)\frac{q^n-p^n}{q-p}y_{i,M}\\
\E_t\pazocal{K}_{t+n} &= p^n k_I + (1-p)\frac{q^n-p^n}{q-p}k_M= \frac{q^{n+1}-p^{n+1}}{q-p}k_I.
\end{align*}
\end{corollary}
\begin{proof}
See Appendix \ref{sec:proof_cor}.
\end{proof}
This corollary establishes that the impulse response function of a typical DSGE model that belongs to class $\mathcal{A}$ can be constructed as a weighted sum of two scalars: the impact Markov state and its medium run counterpart. For the endogenous state, Corollary \ref{cor:IRF} establishes further that one only needs to know the impact Markov state and the endogenous probability of staying in the ``medium run'' state $q$. Likewise, the first two expressions can be used to decompose the impulse response into its impact and medium run effects additively. 

Following the recommendations in both \cite{Uhlig2010} and \cite{Leeper2010}, empirical papers on government spending now routinely report the cumulative or PDV multiplier of government spending \textemdash see \cite{Ramey2018government} for a recent example. I show in the Appendix that, assuming a discount rate of $\beta<1$, this model statistic can be written as in the following corollary:
\begin{corollary}[Present Discount Value]
\label{cor:PDV}
The Present Discount Value multiplier for variable $i$ after a shock $\epsilon_t=1$ is given by:
\begin{align*}
\mathfrak{M}_i=y_{i,I}+\beta\frac{1-p}{1-\beta q}y_{i,M}.
\end{align*}
\end{corollary}
\begin{proof}
See Appendix \ref{sec:proof_cor}.
\end{proof}
In other words, the cumulative/PDV multiplier is just the sum of the short run impact and a second term that is proportional to the medium run effect. Using this formula, it is then straightforward to see whether the potentially negative long run impact outweighs the short run positive impact. If both effects go in the same direction, this formula can also be used to compute the relative contributions of the two effects: for a given PDV multiplier, how much is due to the short run versus medium run effects? Again, one can potentially give a formal answer to this question using Corollary \ref{cor:PDV}. Finally, one can use the latter to make a cleaner comparison between theory and empirical papers as these papers now routinely report such an object that has historically been computed numerically. 

Another model object that is related but somewhat distinct from the PDV is the cumulative sum of a given variable. In some contexts, such a variable is of critical importance. One can think about situations where the cumulative sum of a given variable can be used in order to back out an important variable that is not in the model. In the context of New Keynesian models, these are usually written in terms of inflation with the actual price level left out. Indeed, having the price level brings a unit root to the model and also prevents one from deriving an analytical solution since it adds one endogenous state variable to the model. Given that the expected price level can be constructed using the cumulative sum of inflation rates, the latter is an important object. In this context, I show in the Appendix that the following corollary applies:
\begin{corollary}[Cumulative Sum]
\label{cor:Cumsum}
The cumulative sum for control variable $i$ and the endogenous state after a shock $\epsilon_t=1$ are given by:
\begin{align*}
\mathfrak{C}_{i,y} = \E_t\sum_{n=0}^{\infty}\pazocal{Y}_{i,t+n}&= \frac{1}{1-p}y_{i,I} + \frac{1}{1-q}y_{i,M}\\
\mathfrak{C}_k = \E_t\sum_{n=0}^{\infty}\pazocal{K}_{t+n}
 &=\frac{1}{1-p}k_I + \frac{1}{1-q}k_M= \frac{1}{q(1-q)}k_M
\end{align*}
\end{corollary}
\begin{proof}
See Appendix \ref{sec:proof_cor}.
\end{proof}
Once again, this model statistic is a linear combination of both the short run impact and the medium run effect and can be readily analyzed using standard tools. Each of these is multiplied by the expected number of time periods for the chain to be in each respective state\textemdash conditional on starting in the first state. For the endogenous state, using the implied relationship between these two one can obtain an expression that only depends on the medium run state $k_M$. In a recent contribution, \cite{Nie2022} leverage this corollary to study analytically a policy of Price Level Targeting.\footnote{Given that considering the price level adds an endogenous state variable, papers considering a policy of PLT virtually all report numerical experiments. Using the formula from Corollary \ref{cor:Cumsum}, \cite{Nie2022} derive analytical formulas for a New Keynesian model with a Central Bank that follows a policy of PLT and study how the latter can rid the model of sunspot liquidity traps.}

Finally, the MUMS allows one to produce an exact, two-dimensional representation of specific equations of the model. Importantly, this representation takes endogenous rational expectations into account. This allows one to generalize a tool that has been used extensively for forward-looking models \textemdash see \cite{Eggertsson2010}, but now allowing for endogenous persistence. This will be especially useful for models (such as the one developed in the next section) where most of the action occurs in the medium run. For the sake of space, this is developed in online Appendix A. To illustrate the usefulness of the framework developed in the current paper, I now study a New Keynesian model in which the Central Bank behaves optimally with commitment.

\section{Application: a New Keynesian model with optimal monetary policy}
\label{sec:application}

I consider a standard New Keynesian (NK) model that features optimal monetary policy with commitment as in \cite{Gali2015monetary}, chapter 5. I refer the reader to this textbook for the primitives of the model and focus here on the log-linear version of the equilibrium conditions. Let $x_t$ denote the output gap and $\pi_t$ the inflation rate at time $t$, both in deviations from steady state. The Central Bank seeks to choose a state-contingent sequence $\left\{x_t,\pi_t\right\}_{t=0}^{\infty}$ to solve the following program:
\begin{align}
\nonumber
\min \E_0\sum_{t=0}^{\infty}&\beta^t\left[\pi_t^2+\vartheta x_t^2\right]   \\
\text{s.t}\quad \pi_t =& \beta\E_t\pi_{t+1}+\kappa x_t + u_t 
\label{eq:NKPC}
\end{align}
where $\beta\in(0,1)$ is the discount factor, $\kappa$ is the elasticity of inflation to marginal costs for a given expected inflation rate next period and $u_t$ follows a standard $AR(1)$ process. I first consider the case where the Central Bank is not able to commit. In that case, the Central Bank takes $\E_t\pi_{t+1}$ as given so that the first order conditions boil down to the following targeting rule:
\begin{align}
x_t=-\frac{\kappa}{\vartheta}\pi_t.
\end{align}
In that case there is no endogenous persistence and the dynamics of this economy can be reproduced by a two-state Markov chain where the probability to stay in the first state is $p$ and where the second state is the steady state. It follows that the set of Markov restrictions can be written as:
\begin{align}
\label{eq:NKPC_discretion}
\pi_I^d &= \beta p \pi_I^d + \kappa x_I^d + u_I\\
x_I^d &= -\frac{\kappa}{\vartheta}\pi_I^d,
\label{eq:targ_rule_discretion}
\end{align}
where the superscript $d$ denotes discretion. This linear system can easily be solved\footnote{In that case the solution for inflation is given by:
\[
\pi_I^d = \frac{u_I}{1-\beta p + \kappa^2/\vartheta}u_I.
\]} for $\pi_I^d$ and $x_I^d$ as a function of $u_I$. It is pretty clear that equation \eqref{eq:NKPC_discretion} describes a positive relation between inflation and the output gap. In contrast, the targeting rule \eqref{eq:targ_rule_discretion} describes a negative relation between those same two variables. Note that the cost-push shock will then shift the Phillips curve along the negatively sloping targeting rule: the output gap and inflation will be negatively correlated conditional on a cost-push shock. Because of this, \cite{Mcleay2020optimal} argue that one cannot estimate the slope of the Phillips curve in this situation because one would obtain a negative value which contrasts with the actual value of $\kappa>0$. In the current context, if one were to consider a model with discretion then knowledge of $\pi_I$ and $x_I$ without knowing the shock $u_I$ would lead to estimate a negative slope of the Phillips curve. As in \cite{Mcleay2020optimal}, it will be useful to consider what happens when the Central Bank is able to commit. In this case, the model will feature endogenous persistence and thus the tools developed before will become relevant.

In the case where the Central Bank \textit{is} able to commit, it will take into account the fact that its actions influence $\E_t\pi_{t+1}$. Formally, letting $\xi_t$ denote the Lagrange multiplier in front of the time $t$ constraint for the maximization program, it is well known that the first order conditions boil down to the following two linear difference equations:
\begin{align}
\label{eq:FOC_x}
x_t &= \frac{\kappa}{\vartheta}\xi_t \\
\xi_t &= \xi_{t-1} -\pi_t
\label{eq:FOC_pi}
\end{align}
with $\xi_{-1}=0$.
Following the method described in section \ref{sec:baseline_model}, equations \eqref{eq:NKPC}, \eqref{eq:FOC_x} and \eqref{eq:FOC_pi} can be rewritten with the following Markov restrictions:
\begin{align}
\label{eq:NKPC_I}
\pi_I &=  \beta\left[p\pi_I+(1-p)\pi_M\right]+\kappa x_I+u_I\\   
\label{eq:NKPC_M}
\pi_M &=  \beta q\pi_M + \kappa x_M\\   
\label{eq:LoM_I}
\xi_I &= -\pi_I\\
\label{eq:LoM_M}
\xi_M &= \frac{1}{1-p}\xi_I-\pi_M\\
\label{eq:AR2}
\xi_M &= \frac{q}{1-p}\xi_I,
\end{align}
and where the first order condition with respect to the output gap now becomes $x_i=\frac{\kappa}{\vartheta}\xi_i$ for $i\in\left\{I,M\right\}$. Equations \eqref{eq:NKPC_I}-\eqref{eq:NKPC_M} are the short and medium run versions of the Phillips curve. I have used the Markov structure to compute expected inflation on both right hand sides. Equations \eqref{eq:LoM_I}-\eqref{eq:LoM_M} pertain to the law of motion for the Lagrange multiplier in the short and medium run. Finally, equation \eqref{eq:AR2} encodes the fact that the Lagrange multiplier follows an $AR(2)$ process conditional on the cost push shock. Note that equation \eqref{eq:NKPC_I} with $\pi_M=0$ alongside $x_I=-\frac{\kappa}{\vartheta}\pi_I$ nests the case of full discretion where the Central Bank cannot commit at all. 

From this set of Markov restrictions, one can observe that the dynamics after an increase in $u_I$ might be different in regime $I$ and $M$. Before doing that, it will be useful to solve for $q$, which will govern how long the medium run dynamics will last on average. Combining the medium run Phillips curve \eqref{eq:NKPC_M} with equations \eqref{eq:LoM_M} and \eqref{eq:AR2}, one ends up with a simple quadratic equation in $q$. Moreover, it can be shown that there is only one solution within the unit circle which is given by:
\begin{align*}
q = \frac{2}{\Psi + \sqrt{\Psi^2-4\beta}} ,   
\end{align*}
where $\Psi\equiv 1+\beta + \frac{\kappa^2}{\vartheta}$. As a result, $q$ is clearly decreasing in $\kappa$ and increasing in $\vartheta$ for a given degree of price flexibility.\footnote{If $\vartheta$ is the socially optimal weight, then $\vartheta = \kappa/\theta$, where $\theta$ is the elasticity of substitution across goods. In that case, $q$ is still decreasing with respect to $\kappa$, but is now increasing with respect to $\theta$.} Intuitively, the higher $\kappa$ the more flexible the prices, in which case it is easier for the Central Bank to control inflation. As a result, the Central Bank does not need to keep inflation away from steady state for a long period of time. In the case of stickier prices, it is optimal for the Central Bank to make inflation deviate slightly but keep this deviation for a longer time. 

With this in mind, I will focus first on the short run dynamics. From equation \eqref{eq:NKPC_I}, it is evident that the action comes from the fact that the Phillips curve will shift after an increase in the cost-push shock $u_I$. In addition, expected inflation in the medium run will modify how actual inflation depends on a given variation in output or the cost-push shock. To see this more clearly, notice that one can use equations  \eqref{eq:NKPC_M}, \eqref{eq:LoM_I}, \eqref{eq:AR2} alongside the first order condition with respect to the output gap to write:
\begin{align}
\pi_M = -\frac{\frac{\kappa^2}{\vartheta}}{1-\beta q}\frac{q}{1-p}\pi_I = -\frac{1-q}{1-p}\pi_I, 
\label{eq:piM_piI}
\end{align}
which encodes the promise of lower inflation after the initial increase in the short run. The second equality uses the quadratic polynomial that $q$ has to solve. In addition, $q$ is independent of $p$ and as a result one can see that the more persistent the cost-push shock, the more medium run deflation is promised after a short run increase in inflation. Using this to substitute for $\pi_M$ in the short run Phillips curve \eqref{eq:NKPC_I} and re-arranging, I get:
\begin{align*}
\pi_I &= \beta\overbracket{(p+q-1)\pi_I}^{\E_I \pi_{t+1}} + \kappa x_I+u_I\\
&=    \frac{1}{1-\beta p + \beta (1-q)}\left(\kappa x_I+u_I\right),
\end{align*}
so that the presence of commitment effectively makes the Phillips curve flatter in the short run compared to the case with discretion that obtains under $q=1$. In addition, one can see that the effect of the cost-push shock on expected inflation in the short run is ambiguous: if $p>1-q$, then it is positive. If however the shock is transient then expected deflation in the medium run more than compensates inflation in the short run. 

\begin{figure}
\centering
\caption{Short run equilibrium}
\includegraphics[width=.7\textwidth]{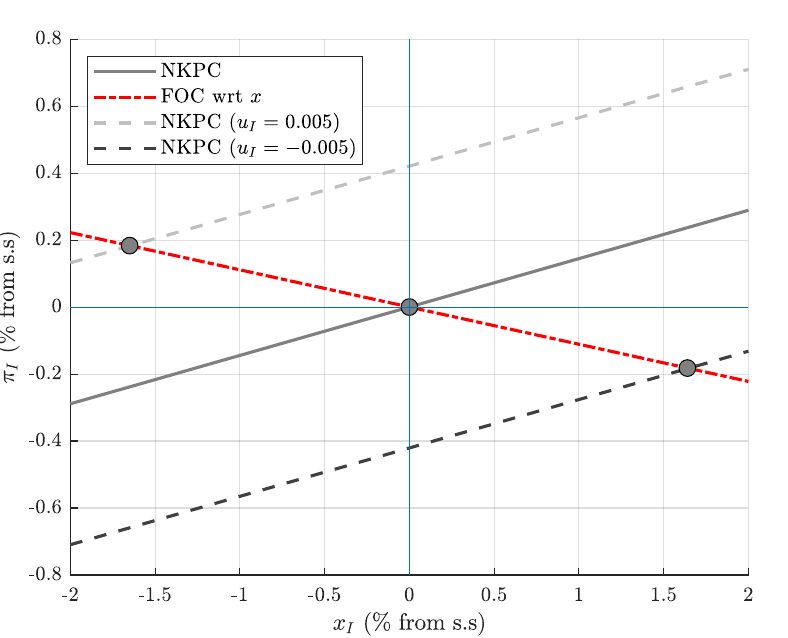}
\vspace{0.1cm}
\begin{minipage}{0.7\textwidth}
\protect\footnotesize Notes : \protect\scriptsize The calibration closely follows \cite{Gali2015monetary}, chapter 3, pp. 67-68. This calibration implies $\kappa = 0.172$. I assume that the Central Bank chooses the optimal weight of $\vartheta=\kappa/\epsilon=0.019$. The shock has persistence $p=0.5$ and magnitude $u_I=0.005$. 
\end{minipage}
\label{fig:Graph_I}
\end{figure}

Combined with equation \eqref{eq:LoM_I} as well as the first order condition with respect to the output gap, one ends up with a set of two equations in $\pi_I,x_I$. These can be plotted graphically and this is done in Figure \ref{fig:Graph_I}. One take-away from that figure is that the economy behaves in exactly the same way in the short run compared to the case of discretion. The only feature that will change is the slope of the Phillips curve plotted in the figure, which will be larger under discretion. With commitment however, the story does not end there and the economy potentially behaves differently in the medium run before returning to steady state. In that situation, the economy can still be represented with a Phillips curve and the first order condition with respect to the output gap. The Phillips curve now takes the following form:
\begin{align}
\pi_M = \frac{\kappa}{1-\beta q}x_M.    
\label{eq:NKPC_M_estim}
\end{align}
Note that the Phillips curve is not shifting anymore in the medium run. In contrast, the targeting rule will shift instead. This is due to the promise made by the Central Bank in the short run. More specifically, the targeting rule can now be written as:
\begin{align}
\pi_M = -\frac{\vartheta}{\kappa}x_M - \frac{\pi_I}{1-p}.    
\label{eq:shift_M}
\end{align}
This explains the finding from before: the more persistent the cost push shock, the larger the shift in the targeting rule and thus the lower inflation has to be in the medium run. These dynamics are illustrated in Figure \ref{fig:Graph_M}, which uses the same calibration as Figure \ref{fig:Graph_I}.
\begin{figure}
\centering
\caption{Medium run equilibrium}
\includegraphics[width=.7\textwidth]{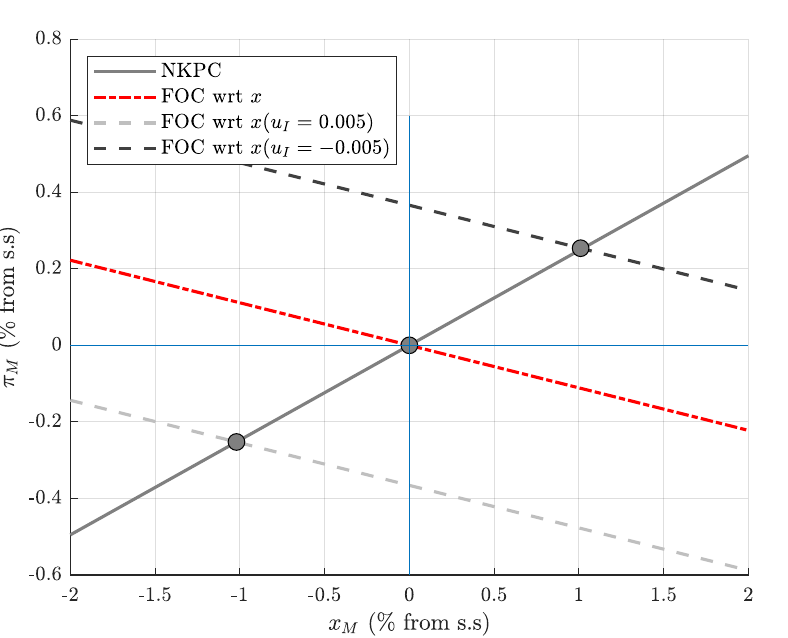}
\label{fig:Graph_M}
\end{figure}
In sharp contrast with short term dynamics, conditional on the cost push shock one now observes a \textit{positive} correlation between inflation and the output gap. Notice that equation \eqref{eq:shift_M} is very similar to equation (11) in \cite{Mcleay2020optimal}. In their case, it arises under discretion and when implementation errors inject a shock in the targeting rule. In the current model, the shift in the targeting rule is endogenous and depends on short run inflation. Therefore, if one has access to $\pi_M$ and $x_M$ then one can readily compute the slope in Figure \ref{fig:Graph_M} using the medium run Phillips curve \eqref{eq:NKPC_M_estim}. 

Given the negative relation between $\pi_M$ and $\pi_I$, the impulse response will be such that it starts out positive and then turns out negative. Indeed, using the results from Section \ref{sec:analytical_tools}, the impulse response at horizon $n$ can be written as:
\begin{align*}
\E_t \Pi_{t+n} = p^n\pi_I -(1-q)\frac{q^n-p^n}{q-p}\pi_I,   
\end{align*}
where I have used the fact that $\pi_M=-(1-q)\pi_I/(1-p)$. The first term on the right hand side initially dominates and is then more than compensated over by the second term. This is further illustrated in Figure \ref{fig:visu_pi}, where I plot the average realization of the Markov chain for inflation alongside selected realizations.
\begin{figure}
\centering
\caption{Impulse response of inflation}
\includegraphics[width=.8\textwidth]{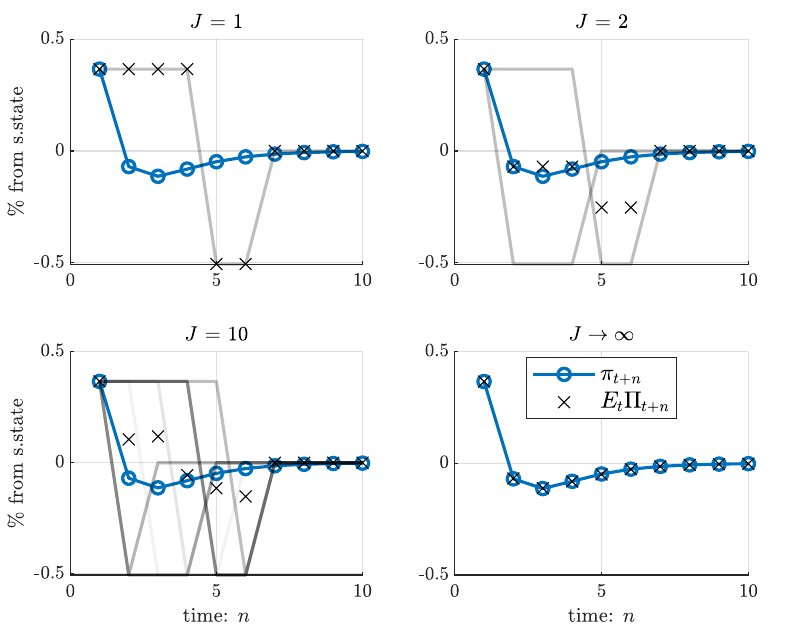}
\label{fig:visu_pi}
\end{figure}
Finally, one can easily compute the expected cumulative sum of inflation to obtain that:
\begin{align*}
\E_t\sum_{n=0}^{\infty}\Pi_{t+n}=\lim_{n\to\infty}\E_t p_{t+n}-p_{-1} &= \frac{1}{1-p}\pi_{I}+\frac{1}{1-q}\pi_{M}\\
&= \frac{1}{1-p}\pi_{I} - \frac{1}{1-q}\frac{1-q}{1-p}\pi_{I}=0,
\end{align*}
which recovers the well-known result that the price level $p_t$ returns to its initial level under full commitment after a cost-push shock. In contrast, the cumulative sum of the output gap is clearly negative given that both $x_I$ and $x_M$ are strictly negative. In some sense, the short term negative relation between output gap and inflation in the short run is perfectly compensated by the positive relation in the medium run.

\subsection{Discussion}

I have chosen this model as the main application given its simplicity, which permits a clear exposition of the method described in the previous sections. However, most linear DSGE models will most likely be more complicated compared to this model. In that case, the analytical expressions will be more involved and thus potentially less useful in aiding the interpretation of the model. However, the tools described in this paper can also be fruitfully used in that case.

To understand how, note that in light of the literature on tractable DSGE models at the lower bound it is now common to solve for one period ahead conditional expectations as a function of current realizations \textemdash see \cite{Eggertsson2010} for an early example. There are no available tools to do that if the model features endogenous persistence however. One alternative is to consider the representation developed in \cite{Sims2001}. For the model considered in the Introduction with a shock $z_t$ that follows an $AR(1)$ with persistence $p$, the variable $y_t$ would have the following law of motion: $y_t = ay_{t-1} + bz_t$, where the parameters $a$ and $b$ need to be solved for. In case $a=0$, then one can write that $\E_t y_{t+1} = py_t$: expectations are endogenous. If $a\neq 0$ however, then the conditional expectation will be written as $\E_ty_{t+1} = a y_t + bpz_t$. One can immediately see that the second expression does not nest the first one in case $a=0$: in that situation, expectations are viewed as purely exogenous. The methods developed in this paper allow one to construct a conditional expectation that \textit{does} nest the simple one.

Returning to the issue at hand, if the model features too many equations so that the analytical solution is not tractable then one can use the tools developed in this paper to represent the model graphically. Importantly, it does so in a manner that does nest the simple case where expectations are viewed as endogenous. To illustrate that, I consider in Online Appendix A a Real Business Cycle (RBC) model with capital adjustment costs and sticky prices. The model nests both the standard RBC and NK models as special cases. I use this model to revisit the early debates surrounding the adoption of the standard NK framework described in the seminal paper of \cite{Gali1999technology}. Following that paper, I focus on the labor market and represent the equilibrium of the model in a labor supply/demand graph and study how a technology shock moves these two curves in different variants of the model. 

The main result from \cite{Gali1999technology} is that after a positive technology shock, hours worked persistently decline. The author then shows that this is inconsistent with an RBC model, but consistent with a simple NK model. In Online Appendix A, I study a version of the NK-RBC model with sticky prices and an arbitrary large capital adjustment cost that nests the standard NK model. In that situation, a positive technology shock makes household supply less labor for a given real wage.\footnote{Given log utility with respect to consumption and a separable, convex disutility associated with hours worked these two are substitutes. For a given real wage, the technology shock makes consumption and thus the marginal rate of substitution increase. In the face of a constant real wage, hours have to decrease.} In contrast, the increase in technology makes firms demand more labor for a given wage. I show that the shift in labor supply dominates so that hours decrease as a result. The intuition behind the result is quite different from the one exposed in \cite{Gali2013notes}, which emphasizes shifts in labor demand in the NK model. 

In order to consider the RBC model, I make the capital adjustment as well as price adjustment costs vanish. In this situation, in line with \cite{Gali1999technology} I find that hours now increase after a technology shock. As before, labor supply decreases while labor demand increases for a given real wage. The absence of accomodative monetary policy after an increase in technology explains why labor supply shifts by less. In addition, the presence of capital enables the household to smooth consumption so that it reacts less in the short run. Furthermore, the addition of capital magnifies the effect on the marginal product of labor and thus magnifies the shift in labor demand. As a result, hours now increase in equilibrium. Importantly however, the presence of investment makes consumption in the short run less sensitive to variations in hours worked so that the slope of the labor supply curves is lower in magnitude. This also helps to explain why hours increase in equilibrium. This effect cannot be captured in standard methods, which would attribute all the effects to shifts without modifying the slopes.

The exercise also highlights that the finding of a decrease in hours in the simple New Keynesian model is rather fragile: the addition of private capital with moderate adjustment costs moves the model closer to the standard RBC. For the calibration that I use, I find that for an elasticity of private investment with respect to Tobin's $q$ that is higher than 3, hours decrease after a technology shock. That is much lower than the value of 10 recently used in \cite{Bilbiie2019capital} or 5.4 in \cite{Bernstein2022simple}. As a result, hours do decrease in the NK model with private capital only if capital adjustment costs are quite large.

Finally, given that the method developed in this paper draws a clear distinction between dynamics in the short and the medium run, I can study the two separately. Using these, I find that the dynamics in the NK-RBC model in the medium run are qualitatively similar to the simple NK model in the short run. In the NK-RBC model, the medium run shifts come from the increase in investment in the short run that feeds into a higher stock of capital. The higher stock of capital increases the marginal product of labor and also permits a higher consumption level. As a result, hours supplied decrease and hours demanded increase. In that case, the decrease in hours supplied dominate so that the real wage increases and hours decrease in the medium run, just like in the standard NK model in the short run. As a result, the effects of an increase in technology in the standard NK model mimic the ones coming from an increase in capital in the medium run in the RBC model.

\section{Conclusion}
\label{sec:conclusion}

In this paper, I have shown that a class of linear DSGE models with one endogenous/exogenous state variables can be equivalently represented by an absorbing 3-state Markov chain. I have used this representation to develop analytical expressions for model objects that have routinely been computed numerically: IRF, PDV multiplier and cumulative sum.

The main limitation of the approach is that it only applies to \textit{linear} DSGE models. As a result, these tools cannot yet be used to study what happens at the Effective Lower Bound or with Downward Nominal Wage Rigidity. Given that the slopes associated with various first order conditions in the simple forward looking NK model have been shown to be crucial to determine policy multipliers, it would be interesting to study how the presence of endogenous persistence modifies these slopes. I leave this avenue for future research.  

\bibliographystyle{apalike2}
\bibliography{MUMS}

\appendix

\section{Proof of Proposition \ref{prop:p_star}}
\label{sec:proof_prop_p_star}

Let $p_{I,t}$ denote the unconditional probability that the Markov chain $\pazocal{K}_t$ is in state `I' at time $t$ and likewise $p_{M,t}$ for state `M'. I do not describe the corresponding probability for state `L' since $y_L=0$. At any point in time $t$, the conditional expectation for the Markov chain $\pazocal{Y}_{t+n}$ is then given by:
\begin{align*}
\E_t  \pazocal{Y}_{t+n} &=  p_{I,t+n}y_I+p_{M,t+n}y_M,  
\end{align*}
for $n\geq 0$, which is basically the law of total probability. Given the transition matrix 
\begin{align*}
\mathcal{P}_b=
\begin{bmatrix}
p & 1-p & 0\\
0 & q & 1-q\\
0 & 0 & 1
\end{bmatrix},
\end{align*}
one can work out the dynamics of the unconditional probabilities $p_{I,t+n}$ and $p_{M,t+n}$. In particular, note that to be in state `I' at time $t+n$, the Markov chain had to be either in state `I' or `M' at time $t+n-1$ for $n\geq 1$. It cannot have been at state `L' since this one is absorbing by construction. Given the transition matrix $\mathcal{P}_b$, the following set of recursions must hold:
\begin{align*}
p_{I,t+n} &= pp_{I,t+n-1}\\    
p_{M,t+n} &= (1-p)p_{I,t+n-1}+qp_{M,t+n-1}
\end{align*}
for $n\geq 1$ and where $p_{I,t}, p_{M,t}$ are given at time $n=0$. The transition matrix for transient states associated with matrix $\mathcal{P}_b$ is:
\begin{align*}
\mathcal{Q} =
\begin{bmatrix}
p & 1-p\\
0 & q
\end{bmatrix}.
\end{align*}
With this in mind, the recursions can be written in matrix form as follows:
\begin{align}
\label{eq:recursion_Q}
\begin{bmatrix}
p_{I,t+n} \\
p_{M,t+n} 
\end{bmatrix}
&=
\mathcal{Q}^\top
\begin{bmatrix}
p_{I,t+n-1} \\
p_{M,t+n-1} 
\end{bmatrix}
\end{align}
Using the definition of the conditional expectation as well as the recursion for the probabilities, one can write for $n\geq 1$:
\begin{align}
\E_t \pazocal{Y}_{t+n} &= \left[y_I\ y_M\right]   
\begin{bmatrix}
p_{I,t+n} \\
p_{M,t+n} 
\end{bmatrix}= \left[y_I\ y_M\right]\mathcal{Q}^\top
\begin{bmatrix}
p_{I,t+n-1} \\
p_{M,t+n-1} 
\end{bmatrix}
\label{eq:Kt_ptm1}
\end{align}
Now using the Cayley-Hamilton theorem and given that $\mathcal{Q}^\top$ is invertible, one can write:
\begin{align*}
\mathcal{Q}^\top = \Tr\left(\mathcal{Q}^\top\right)I_2  - \det \left(\mathcal{Q}^\top\right) \left(\mathcal{Q}^\top\right)^{-1}
\end{align*}
Using this to substitute for $\mathcal{Q}^\top$ in equation \eqref{eq:Kt_ptm1}, one obtains:
\begin{align}
\nonumber
\langle  \pazocal{Y}_{t+n}\rangle &=  \left[y_I\ y_M\right] \left(  \Tr\left(\mathcal{Q}^\top\right)I_2  - \det \left(\mathcal{Q}^\top\right) \left(\mathcal{Q}^\top\right)^{-1}\right)
\begin{bmatrix}
p_{I,t+n-1} \\
p_{M,t+n-1} 
\end{bmatrix}\\
\nonumber
&= \Tr\left(\mathcal{Q}^\top\right)\left[y_I\ y_M\right]\begin{bmatrix}
p_{I,t+n-1} \\
p_{M,t+n-1} 
\end{bmatrix} - \det \left(\mathcal{Q}^\top\right)\left[y_I\ y_M\right]\left(\mathcal{Q}^\top\right)^{-1}\begin{bmatrix}
p_{I,t+n-1} \\
p_{M,t+n-1} 
\end{bmatrix}\\
&= \Tr\left(\mathcal{Q}^\top\right)\E_t  \pazocal{Y}_{t+n-1} - \det \left(\mathcal{Q}^\top\right)\left[y_I\ y_M\right]\left(\mathcal{Q}^\top\right)^{-1}\begin{bmatrix}
p_{I,t+n-1} \\
p_{M,t+n-1} 
\end{bmatrix}.
\label{eq:matrix_recursion2}
\end{align}
Now notice that equation \eqref{eq:recursion_Q} can be re-written as:
\begin{align*}
\left(\mathcal{Q}^\top\right)^{-1}
\begin{bmatrix}
p_{I,t+n} \\
p_{M,t+n} 
\end{bmatrix}
&=
\begin{bmatrix}
p_{I,t+n-1} \\
p_{M,t+n-1} 
\end{bmatrix}
\end{align*}
for $n\geq 1$. Using this to simplify the second term on the right hand side of equation \eqref{eq:matrix_recursion2}, one obtains:
\begin{align*}
\E_t  \pazocal{Y}_{t+n} &=    \Tr\left(\mathcal{Q}^\top\right)\E_t  \pazocal{Y}_{t+n-1} - \det \left(\mathcal{Q}^\top\right)\left[y_I\ y_M\right]\begin{bmatrix}
p_{I,t+n-2} \\
p_{M,t+n-2} 
\end{bmatrix} \\
&= \Tr\left(\mathcal{Q}^\top\right)\E_t \pazocal{Y}_{t+n-1} - \det \left(\mathcal{Q}^\top\right)\E_t  \pazocal{Y}_{t+n-2}
\end{align*}
for $n\geq 2$. 
As a result, one obtains:
\begin{align*}
\E_t  \pazocal{Y}_{t+n} &= \left(p+\eta_{kk}\right)\E_t\pazocal{Y}_{t+n-1} -p\eta_{kk}\E_t\pazocal{Y}_{t+n-2}   ,
\end{align*}
for $n\geq 2$, which is the same recursion as the one obtained for $y_{t+n}$ in the main text. 
Given the initial distribution, the Markov chain has to start at state 1 so that $\E_t  \pazocal{Y}_{t}= y_I $ for $n=0$. It follows that
\[
\E_t \pazocal{Y}_{t+1} = p\cdot y_I + p_{12}\cdot y_M
\]
for $n=1$.

\section{Proof of Theorem \ref{thm:IRF_univariate}}
\label{sec:Proof_thm_univariate}

Let us begin with the auto-regressive model, which I reproduce here for convenience:
\begin{align}
\label{eq:fwd_multi_I1_app}
\mathbf{A}_0\mathbb{Y}_t &= \mathbf{A}_1\E_t \mathbb{Y}_{t+1}+B_0 k_t+C_0z_t\\
\label{eq:bwd_multi_I1_app}
k_t &= \rho k_{t-1} + D_0 \mathbb{Y}_t+ez_t\\
z_t &= pz_{t-1}+\epsilon_t.
\label{eq:exo_state_multi_I1_app}
\end{align}
Let me guess that the solution can be written as
\begin{align*}
\mathbb{Y}_t &= M_k k_{t-1} + M_z z_t\\
k_t &= \eta_k k_{t-1} + \eta_z z_t.
\end{align*}
Plugging these guesses into equation \eqref{eq:fwd_multi_I1_app}, I get:
\begin{align*}
\mathbf{A}_0\left[M_k k_{t-1} + M_z z_t\right] &= \mathbf{A}_1\E_t\left[M_k k_{t} + M_z z_{t+1}\right]+B_0\left[\eta_k k_{t-1} + \eta_z z_t\right]+C_0z_t\\
&= \left[\eta_k\mathbf{A}M_k+B_0\eta_k\right]k_{t-1}+\left[\eta_z\mathbf{A}M_k+p\mathbf{A}M_z+B_0\eta_z+C_0\right]z_t.
\end{align*}
Now plugging these guesses into equation \eqref{eq:bwd_multi_I1_app}, I obtain:
\begin{align*}
\eta_k k_{t-1} + \eta_z z_t &= \rho k_{t-1} + D_0 \left[M_k k_{t-1} + M_z z_t\right]+ez_t  \\
&= \left[\rho+D_0M_k\right]k_{t-1} + \left[D_0M_z+e\right]z_t.
\end{align*}
By identification, I get that $\eta_k$ and $M_k$ solve the following system of equations:
\begin{align}
\label{eq:UC_1}
\mathbf{A}_0M_k &= \eta_k\mathbf{A}M_k+B_0\eta_k\\
\eta_k &= \rho+D_0M_k
\label{eq:UC_2}
\end{align}
Note that equation \eqref{eq:UC_1} can be rewritten as
\begin{align*}
M_k = \left(\mathbf{A}_0-\eta_k\mathbf{A}_1\right)^{-1}B_0\eta_k.    
\end{align*}
Plugging that back into equation \eqref{eq:UC_2} I finally get the following implicit equation for $\eta_k$:
\begin{align*}
\eta_k = \rho + \eta_kD_0 \left(\mathbf{A}_0-\eta_k\mathbf{A}_1\right)^{-1}B_0.   
\end{align*}

The goal now is to show that parameter $q$ from model $\mathcal{M}$ is a solution to the same equation. A solution of model $\mathcal{M}$ is a vector $\left[q\ k_I\ k_M\ y_{1,I}\dots\ y_{N,I}\ y_{1,M}\dots\ y_{N,M}\right]^\top$ of size $2N+3$ that solves the following system of non-linear equations:
\begin{align}
\label{eq:fwd_I_multi_app}
\mathbf{A}_0\pazocal{Y}_I &= \mathbf{A}p\pazocal{Y}_I+(1-p)\mathbf{A}_1\pazocal{Y}_M + B_0k_I+C_0\\
\label{eq:fwd_M_multi_app}
\mathbf{A}_0\pazocal{Y}_M &= \mathbf{A}q\pazocal{Y}_M +B_0k_M\\
\label{eq:bwd_I_multi_app}
k_I &= D_0\pazocal{Y}_I+e\\
\label{eq:bwd_M_multi_app}
k_M &= \rho \frac{k_I}{1-p} + D_0\pazocal{Y}_M\\
\label{eq:kMkI_multi_app}
k_M &= \frac{qk_I}{1-p}
\end{align}
Using equation \eqref{eq:fwd_M_multi_app}, one can write
$\pazocal{Y}_M = \left(\mathbf{A}_0-q\mathbf{A}_1\right)^{-1}B_0k_M$.
Combining equations \eqref{eq:bwd_M_multi_app}-\eqref{eq:kMkI_multi_app}, one obtains:
\begin{align*}
qk_M &= \rho k_M + qD_0\pazocal{Y}_M = \rho k_M +qD_0\left(\mathbf{A}_0-q\mathbf{A}_1\right)^{-1}B_0k_M
\end{align*}
Assuming that $k_M\neq 0$ and then dividing both sides of this equation by $k_M$, one obtains the same implicit equation as before. It follows that one can choose the unique value of $q$ using the procedure outlined in \cite{Mccallum1983non}. What is now left to show is the the Markov chains actually solve equations \eqref{eq:fwd_multi_I1_app}-\eqref{eq:bwd_multi_I1_app}. Note that equations \eqref{eq:fwd_I_multi_app}-\eqref{eq:fwd_M_multi_app} imply that
\be 
\mathbf{A}_0\E_t\pazocal{Y}_{i,t+n} = \mathbf{A}_1\E_{t+n}\pazocal{Y}_{i,t+n+1} + B_0\E_t\pazocal{K}_{t+n}+C_0\E_t\pazocal{Z}_{t+n}
\label{eq:induction_forward}
\ee
for $n=0,1$. To show that it holds for $n\geq 2$, I will proceed by induction. Accordingly, assume that equation \eqref{eq:induction_forward} holds for the time periods $t+n$ and $t+n+1$ for $n\geq 2$. What is left is to show that this implies that \eqref{eq:induction_forward} holds for time period $t+n+2$. Using the recursive representation from Proposition \ref{prop:p_star}, one can write;
\begin{align*}
\mathbf{A}_0\E_t \pazocal{Y}_{i,t+n} = (p+q)\mathbf{A}_0\E_t\pazocal{Y}_{i,t+n-1}-pq\mathbf{A}_0\E_t\pazocal{Y}_{i,t+n-2}    
\end{align*}
for $n\geq 2$ and likewise for the Markov chains governing the endogenous and exogenous states. Therefore, one can write:
\begin{align*}
\mathbf{A}_0\E_t \pazocal{Y}_{i,t+n+2} &=(p+q)\mathbf{A}_0\E_t\pazocal{Y}_{i,t+n+1}-pq\mathbf{A}_0\E_t\pazocal{Y}_{i,t+n}\\
&=(p+q)\mathbf{A}_1\E_t\pazocal{Y}_{i,t+n+2}-pq\E_t\mathbf{A}_1\pazocal{Y}_{i,t+n+1}   \\
&+ (p+q)B_0\E_t\pazocal{K}_{i,t+n+1}-pqB\E_t\pazocal{K}_{i,t+n}\\
&+ (p+q)C_0\E_t\pazocal{Z}_{i,t+n+1}-pqC\E_t\pazocal{Z}_{i,t+n}\\
&= \mathbf{A}_1\E_t\pazocal{Y}_{i,t+n+3}+B_0\E_t\pazocal{K}_{i,t+n+2}+C_0\E_t\pazocal{Z}_{i,t+n+2}
\end{align*}
As a result, the induction hypothesis implies that the forward equation holds for time period $t+n+2$ as well. Hence proven. A similar logic applies for the backward equation as well as the law of motion for the exogenous state.

\section{Proofs of Corollaries in Section \ref{sec:analytical_tools}}
\label{sec:proof_cor}

\noindent \textbf{Proof of Corollary \ref{cor:IRF}}.
Let me start from the transition matrix $\mathcal{P}_b$. This matrix can be written in canonical form as follows:
\[
\mathcal{P}_b=
\left(\begin{array}{@{}c|c@{}}
\mathcal{Q}
  & R \\
\hline
  0_{1\times 2} &
\mathbf{I}\end{array}\right)
\]
where $\mathcal{Q} \in \mathbb{R}^{2\times 2}$, $R \in \mathbb{R}^{2\times 1}$ and $\mathbf{I}$ is a unit scalar matrix. I focus the case where $p\neq q$. The special case $p=q$ can be handled in a similar way. I start by showing by induction that the $n-$th power of the transition matrix can be written as
\[
(\mathcal{P}_b)^n=
\left(\begin{array}{@{}c|c@{}}
\mathcal{Q}^n
  & * \\
\hline
  0_{1\times 2} &
\mathbf{I}\end{array}\right)
\]
where the $*$ block is left unspecified. For $n=1$, this is trivial. Assume now that this relationship is true for all $n\geq 1$. Then, we can express $(\mathcal{P}^*)^{n+1}$ as
\begin{align*}
(\mathcal{P}_b)^{n+1}&=
\left(\begin{array}{@{}c|c@{}}
\mathcal{Q}
  & R \\
\hline
  0_{1\times 2} &
\mathbf{I}\end{array}\right)
\left(\begin{array}{@{}c|c@{}}
\mathcal{Q}^n
  & * \\
\hline
 0_{1\times 2} &
\mathbf{I}\end{array}\right)=
\left(\begin{array}{@{}c|c@{}}
\mathcal{Q}^{n+1}
  & * \\
\hline
  0_{1\times 2} &
\mathbf{I}\end{array}\right)
 \end{align*}
which proves the result for all $n\geq 1$. I will now use induction again to show that the $n-$th power of $\mathcal{Q}$ can be expressed as:
\[
\mathcal{Q}^{n} = 
\begin{pmatrix}
p^n & (1-p)\frac{q^n-p^n}{q-p} \\
0 & q^n \\
\end{pmatrix}
\]
For $n=1$ this is trivial. Assume now that this relationship is true for all $n\geq 1$. Then, we can express $\mathcal{Q}^{n+1}$ as:
\begin{align*}
\mathcal{Q}^{n+1}&=
\mathcal{Q}^{n}\mathcal{Q}=
\begin{pmatrix}
p^n & (1-p)\frac{q^n-p^n}{q-p} \\
0 & q^n \\
\end{pmatrix}
\begin{pmatrix}
p & 1-p \\
0 & q
\end{pmatrix}
 \end{align*}
It is clear that the diagonal elements are given by $p^{n+1}$ and $q^{n+1}$ respectively. For the upper-right term, we have:
\begin{align*}
p^n(1-p)+(1-p)q\frac{q^n-p^n}{q-p}&=(1-p)\left[p^n+q\frac{q^n-p^n}{q-p}\right]\\
&= (1-p)\left[p^n\left(1-\frac{q}{q-p}\right)+\frac{q^{n+1}}{q-p}\right]\\
&= (1-p)\left[p^n\frac{q-p-q}{q-p}+\frac{q^{n+1}}{q-p}\right]\\
&= (1-p)\frac{q^{n+1}-p^{n+1}}{q-p},
\end{align*}
which proves the desired result. For the expression regarding the endogenous state/backward-looking variable, note that:
\begin{align*}
\E_t\pazocal{K}_{t+n} &= p^nk_I + (1-p)\frac{q^n-p^n}{q-p}k_M\\
&= p^nk_I + q\frac{q^{n}-p^n}{q-p}k_I\\
&= p^nk_I + \frac{q^{n+1}-qp^n}{q-p}k_I\\
&= \left[p^n+\frac{q^{n+1}-qp^n}{q-p}\right]k_I\\
&= \left[\frac{p^n(q-p)}{q-p}+\frac{q^{n+1}-qp^n}{q-p}\right]k_I\\
&= \frac{q^{n+1}-p^{n+1}}{q-p}k_I,
\end{align*}
where I have used the fact that $k_M=qk_I/(1-p)$ on the second line. 

\noindent \textbf{Proof of Corollary \ref{cor:PDV}}.
Using the expression for the impulse response of $\pazocal{Y}_{i,t}$ from Corollary \ref{cor:IRF}, the expected cumulative discounted sum starting from time $t$ is defined as
\begin{align*}
\E_t\sum_{n=0}^{\infty}\beta^n \pazocal{Y}_{i,t+n} &=
 \sum_{n=0}^{\infty}\beta^n \E_t\pazocal{Y}_{i,t+n}\\
&=\sum_{n=0}^{\infty}\left\{\beta^n p^n y_{i,I}+\beta^n(1-p)\frac{q^n-p^n}{q-p}y_{i,M}\right\}\\
						&= \sum_{n=0}^{\infty}(\beta p)^ny_{i,I}+\frac{1-p}{q-p}y_{i,M}\sum_{n=0}^{\infty}(\beta)^n(q^n-p^n)\\
						&= \frac{y_{i,I}}{1-\beta p }+\frac{1-p}{q-p}y_{i,M}\left(\frac{1}{1-\beta q}-\frac{1}{1-\beta p}\right)\\
						&= \frac{y_{i,I}}{1-\beta p }+\beta\frac{1-p}{(1-\beta q)(1-\beta p)}y_{i,M}
\end{align*}
Therefore, the present discount value multiplier for variable $i$ can be written as
\[
\E_t\sum_{n=0}^{\infty}\beta^n \pazocal{Z}_{t+n} = \frac{1}{1-\beta p }\quad\Rightarrow\quad \mathfrak{M}_i = \frac{\frac{y_{i,I}}{1-\beta p }+\beta\frac{1-p}{(1-\beta q)(1-\beta p)}y_{i,M}}{\frac{1}{1-\beta p }}.
\]
and then multiplying both numerator and denominator by $(1-\beta q)(1-\beta p)$.

\noindent \textbf{Proof of Corollary \ref{cor:Cumsum}}.
Using the expression for the impulse response of $\pazocal{K}_t$ from Corollary \ref{cor:IRF}, the conditional expected cumulative sum can be expressed as: \begin{align*}
\E_t\sum_{n=0}^{\infty}\pazocal{K}_{t+n} &= \E_t\sum_{n=0}^{\infty} \bigg [p^n k_I+(1-p)\frac{q^n-p^n}{q-p} k_M \bigg] \\
                                &= \frac{1}{1-p} k_I+\frac{1-p}{q-p} \sum_{n=0}^{\infty} \bigg ( q^n-p^n\bigg)k_M \\
                                &= \frac{1}{1-p} k_I+\frac{1-p}{q-p}  \bigg ( \frac{1}{1-q}-\frac{1}{1-p}\bigg)k_M\\
                                &= \frac{1}{1-p}k_I + \frac{1}{1-q}k_M.
\end{align*}
A similar expression can be obtained for the forward-looking variable. Using the fact that 
\[
k_M = \frac{qk_I}{1-p}
\]
and re-arranging, one obtains the expression in the main text.

\end{document}